\documentclass[11pt,a4paper]{article}

\usepackage[utf8]{inputenc} 
\usepackage[T1]{fontenc}    %
\usepackage[english]{babel} 
\usepackage[final]{microtype} 

\usepackage{amssymb,amsmath,mathtools} 
\usepackage{graphicx} 
\usepackage{bbm} 
\usepackage{bm} 
\usepackage{tikz} 

\usepackage[font=small,labelfont=bf]{caption}

\usepackage[hmargin=0.12\paperwidth,vmargin=0.2\paperwidth,bindingoffset=0cm]{geometry}

\usepackage{amsthm}

\theoremstyle{plain}
  \newtheorem{theorem}{Theorem}[section]
  \newtheorem{proposition}[theorem]{Proposition}
  \newtheorem{lemma}[theorem]{Lemma}
  \newtheorem{corollary}[theorem]{Corollary}
  
\theoremstyle{definition}
  \newtheorem{definition}[theorem]{Definition}

\theoremstyle{remark}

\newcommand{\nn}{\nonumber} 

\newcommand{\N}{\mathbb{N}} 

\newcommand{\R}{\mathbb{R}}

\newcommand{\E}{\mathbb E}

\newcommand{\prob}{\mathbb P}

\renewcommand{\phi}{\varphi} 
\renewcommand{\epsilon}{\varepsilon} 

\DeclareMathOperator{\tr}{Tr} 
\DeclareMathOperator{\diag}{diag} 

\DeclareMathOperator{\Ai}{Ai} 

\newcommand{\Her}{\mathcal H}

\usepackage{hyperref}   
\hypersetup{
 pdfborder={0 0 0},      
 colorlinks=true,        
 linktoc=page,           
 linkcolor=blue,         
 citecolor=blue,         
}

\title{\bfseries\Large The Laguerre Unitary Process}
\author{J.~R.~Ipsen \\[1em]%
\small ARC Centre of Excellence for Mathematical and Statistical Frontiers,\\%
\small School of Mathematics and Statistics, The University of Melbourne, Victoria 3010, Australia%
}
\date{\today}

\begin{document}	

\maketitle

\begin{abstract}
\noindent
We define a new matrix-valued stochastic process with independent stationary increments from the Laguerre Unitary Ensemble, which in a certain sense may be considered a matrix generalisation of the gamma process. We show that eigenvalues of this matrix-valued process forms a spatiotemporal determinantal point process and give an explicit expression for the correlation kernel in terms of Laguerre polynomials. Furthermore, we show that in an appropriate long time scaling limit, this correlation kernel becomes identical to that of Dyson Brownian Motion.
\end{abstract}

\section{Introduction}

One of the most studied random matrix ensembles is the Laguerre Unitary Ensemble (LUE); often also referred to as the Wishart Ensemble. In this paper, we construct a new matrix-valued stochastic process with stationary and independent LUE distributed increments. We will refer to this process as the Laguerre Unitary Process (LUP), but it is important to emphasise that this process is very different from what have been called the Wishart/Laguerre Process in~\cite{Bru1991,KO2001}. Slightly oversimplified, the processes considered in~\cite{Bru1991,KO2001} may be thought of as matrix generalisations of the (squared) Bessel Process, while the process considered in this paper is more appropriately thought of as matrix a generalisation of the Gamma Process. We may also think of the LUP an additive version of the multiplicative process recently introduced in~\cite{Strahov2015}, which was based on recent progress regarding products of random matrices~\cite{AKW2013,AIK2013}, see also the review~\cite{AI2015}. Our construction of the LUP is also inspired by recent progress on sum and products of random matrices, in particular the paper~\cite{KR2016}. It also worth mentioning that the LUP has certain similarities with the LUE projection process~\cite{FN2011} closely related to the GUE minor process~\cite{JN2006}.

We will show that the eigenvalues of the LUP form a spatiotemporal (or multi-level) determinantal point process. Determinantal point processes plays an important role in many problems arising in Mathematical Physics, see e.g. the reviews~\cite{Soshnikov2000,Johansson2005,Konig2005,BG2016}. The most well-known example of a \emph{spartiotemporal} determinantal point process is the ($\beta=2$) Dyson Brownian motion~\cite{Dyson1962}, see~\cite{Katori2016} for a review. Dyson Brownian Motion itself may be thought of as arising from the eigenvalues of a matrix-valued stochastic process with stationary and independent increments from the Gaussian Unitary Ensemble (GUE). It should be clear that this is very similar to our construction of the LUP. It would be interesting to study matrix-valued processes with stationary and independent increments which gives rise to spatiotemporal determinantal point processes in higher generality. We note that such generalisations are different from the constructing of matrix-valued L\'evy processes as we do not restrict our attention to continuous-time processes and, thus, it is meaningless to speak of continuity in probability as is done for L\'evy processes (we refer to~\cite{Applebaum2009} for more on L\'evy processes). In fact, it is part of our claim that it is worthwhile not to impose such restriction since there are interesting matrix-valued processes with stationary independent increments which cannot be defined for continuous-time. Case in point, we will define the LUP as a discrete-time process and it does not have an extension to continuous-time as we will discuss below. 

Let us first define the LUE. Throughout this paper, $\Her_N\simeq \R^{N^2}$ denotes the space of $N\times N$ Hermitian matrices. 

\begin{definition}\label{def:LUE}
The Laguerre Unitary Ensemble (LUE) is an absolute continuous $\Her_N$-valued random variable with probability density function (PDF)
\begin{equation}
\mathcal P_{a,b}(X):=
\begin{cases}
\displaystyle\frac1{\mathcal Z_{a,b}(N)} e^{-b\tr X}\det(X)^{a} & \text{if $X$ is positive definite}\\
0 & \text{otherwise}
\end{cases} \label{Wishart}
\end{equation}
where $a>-1$ and $b>0$ are real constants and $\mathcal Z_{a,b}(N)$ is a known normalisation constant. We will sometimes refer to such a random variable as a LUE random matrix.
\end{definition}

We note that for $N=1$, a LUE random matrix is the familiar gamma distributed random variable.
Furthermore, if $a$ is a non-negative integer and $\mathbf X$ is an $N\times (N+a)$ random matrix those entries are i.i.d. centred complex Gaussians with variance $1/b$, then $\mathbf Y:=\mathbf X\mathbf X^*$ is a LUE random matrix with PDF $\mathcal P_{a,b}$. Here, $\mathbf X^*$ denotes the Hermitian transpose of $\mathbf X$. The Wishart/Laguerre process mentioned earlier and considered in~\cite{Bru1991,KO2001} is constructed by letting the entries of the random matrix $\mathbf X$ be i.i.d. real/complex Brownian motions.
Our construction of the LUP is very different and is based on the following observation:

\begin{theorem}\label{thm:sum}
If $\mathbf X$ and $\mathbf Y$ are independent LUE random matrices with PDFs $\mathcal P_{a,b}$ and $\mathcal P_{a',b}$, then $\mathbf X+\mathbf Y$ is a LUE random matrix with PDF $\mathcal P_{a+a'+N,b}$.
\end{theorem}

For $N=1$, Theorem~\ref{thm:sum} is a well-known property for gamma distributed random variables. We will present a proof of Theorem~\ref{thm:sum} very similar to working presented in~\cite{KR2016}, but we will need some facts introduced in Section~\ref{sec:eigenvalues} and we therefore postpone the proof to Appendix~\ref{appendixA}. 

\begin{definition}\label{def:LUP}
The Laguerre Unitary Process (LUP) $\{\mathbf L(t)\}_{t=0,1,\ldots}$ is a $\Her_N$-valued discrete-time stochastic process defined as
\begin{equation}\label{LUP}
\mathbf L(t)-\mathbf L(t-1)=\mathbf X(t)\quad \text{for}\quad 
t=1,2,3,\ldots \quad\text{and}\quad \mathbf L(0)=0,
\end{equation}
where $\{\mathbf X(t)\}_{t}$ are i.i.d. LUE random matrices with PDF $\mathcal P_{0,1}$. 
\end{definition}

The following corollary follows straightforwardly.

\begin{corollary}\label{cor:properties}
The LUP have the following properties:
\begin{enumerate}
 \item[(i)] \emph{Independent increments:} Let $\infty>t_n>\cdots>t_1>t_0\geq0$ be integers. The random matrices
 $\{\mathbf L(t_k)-\mathbf L(t_{k-1})\}_{k=1,\ldots,n}$ are independent. 
 \item[(ii)] \emph{Stationary increments:} Let $\infty> t>s\geq0$ and $\infty>u\geq 0$ be integers. The random matrices 
 $\mathbf L(t)-\mathbf L(s)$ and $\mathbf L(t+u)-\mathbf L(s+u)$ are equal in distribution.
 \item[(iii)] \emph{Unitary invariance:} Let $\infty> t>0$ be an integer and let $U\in U(N)$ be an $N\times N$ unitary matrix. The random matrices $\mathbf L(t)$ and $U\mathbf L(t)U^{-1}$ are equal in distribution.
 \item[(iv)] \emph{LUE marginal distribution:} Let $\infty> t>0$ be an integer. The random matrix $\mathbf L(t)$ has PDF $\mathcal P_{N(t-1),1}$ given by~\eqref{Wishart}.
\end{enumerate}
\end{corollary}

As final remark of this section, it is worth to briefly mention why we have defined the LUP as a \emph{discrete-time} process.
Imagine that we want to construct stochastic process with independent stationary LUE increments $\{\widetilde{\mathbf L}(t)\}_t$. Let $t,s\in\mathcal T$ be positive times and $\mathcal T\subseteq\R_+$ be some index set. Up to an overall multiplicative constants and a rescaling of time, Theorem~\ref{thm:sum} tells us that we must have
\begin{equation}\label{step}
\prob\Big(\widetilde{\mathbf L}(t+s)-\widetilde{\mathbf L}(t)\in[X,X+dX)\Big)=\mathcal P_{N(s-1),1}(X)dX.
\end{equation}
Now, we note that the constraint on the index $a$ in the PDF~\eqref{Wishart} implies a constraint on the time steps $s>1-\frac1N$. Thus, it is impossible contruct a continuous-time process with independent stationary LUE increments for $N\geq 2$. Moreover, typical interest regards the large-$N$ limit, which implies the constraint $s\geq 1$, which is the reason why we have chosen $\mathcal T=\mathbb N$ in Definition~\ref{def:LUP}. For $N=1$, our time step constraint reads $s>0$ and the well-known Gamma Process is a continuous-time process with independent stationary $N=1$ LUE (i.e. gamma distributed) increments.

\section{Eigenvalues of the LUP}
\label{sec:eigenvalues}

As alluded to in the Introduction, our main interest is typically not random matrices themselves but rather their eigenvalues. 
It is a classical result from random matrix theory (see e.g.~\cite{Forrester2010}), that the PDF~\eqref{Wishart} for the LUE induces a joint PDF on the eigenvalues given by
\begin{equation}\label{LUE}
P_{a,b}(x_1,\ldots,x_N):=
\frac1{Z_{a,b}(N)}\triangle_N(x)^2\prod_{i=1}^N w_{a,b}(x_i),\qquad x_1,\ldots,x_N\in\R
\end{equation}
where 
\begin{equation}\label{normalisation}
Z_{a,b}(N)=\frac{N!}{b^{N(N-1)}}\prod_{j=1}^N\frac{\Gamma(j)\Gamma(a+j)}{\Gamma(a+1)}
\end{equation}
is a normalisation constant, 
\begin{equation}\label{weight}
w_{a,b}(x):=
\begin{cases}
 \displaystyle\frac{b^{a+1}}{\Gamma(a+1)}x^{a}e^{-bx} & \text{if }x>0\\
 0 & \textup{otherwise}
\end{cases}
\end{equation}
is a weight function, and
\begin{equation}
\triangle_N(x):=\det_{1\leq k,\ell\leq N}[x_k^{\ell-1}]=\prod_{1\leq k<\ell\leq N}(x_k-x_\ell)
\end{equation}
is the so-called Vandermonde determinant. We note that the weight function~\eqref{weight} is identical to the PDF of a gamma distributed random variable.

It is our goal to study the eigenvalues of the LUP as a stochastic process on $\R^N$ (we do not order the eigenvalues). 
Let $\wp_t$ denote the joint PDF of the eigenvalues of $\mathbf L(t)$. From Corollary~\ref{cor:properties} (iv) we know that $\mathbf L(t)$ is a random matrix from the LUE with density $\mathcal P_{N(t-1),1}$, and therefore that
\begin{equation}\label{jpdf}
\wp_t(x_1,\ldots,x_N)=P_{N(t-1),1}(x_1,\ldots,x_N),
\end{equation}
where $P_{a,b}$ is given by~\eqref{LUE}. Our first main result will be to determine the transition probability density $p_{t,s}$ such that
\begin{equation}
\wp_t(y_1,\ldots,y_N)=\int_{\R^N}p_{t,s}(y_1,\ldots,y_N\vert x_1,\ldots,x_N)\wp_s(x_1,\ldots,x_N)dx_1,\ldots,dx_N
\end{equation}
for $t>s>0$.

\begin{theorem}
Let $t>s>0$. The transition probability density for the eigenvalues of the LUP is given by
\begin{equation}\label{tranistion}
 p_{t,s}(y_1,\ldots,y_N\vert x_1,\ldots,x_N)=
 \frac1{N!}\frac{\triangle_N(y)}{\triangle_N(x)}\det_{1\leq k,\ell\leq N}\Big(\kappa_{t-s}(y_k-x_\ell)\Big),
\end{equation}
where $\kappa_{t-s}(x)=w_{N(t-s)-1,1}(x)$ and $w_{a,b}$ is defined by~\eqref{weight}.
\end{theorem}

\begin{proof}
Let $X=\diag(x_1,\ldots  x_N)$. We know from~\eqref{step} that we have
 \begin{equation}
 \prob\Big(\mathbf L(t)\in [H,H+dH)\Big\vert \mathbf L(s)=X\Big)
 =\mathcal P_{N(t-s-1),b}(H-X)dH
 \end{equation}
and from~\eqref{Wishart} that
 \begin{equation}
  \mathcal P_{N(t-s-1),b}(H-X)=
  \frac1{\mathcal Z_{N(t-s-1),1}(N)}\prod_{n=1}^Ne^{x_n} e^{-\tr H}{\det}_+(H-X)^{N(t-s-1)}
 \end{equation}
where
 \begin{equation}
 {\det}_+(A)=
 \begin{cases}
  \det(A) & \text{if $A$ is positive definite}\\
  0 & \text{otherwise}
 \end{cases}.
 \end{equation}
To find the transition probability density for the eigenvalues, we must make an eigenvalue decomposition for $H$ and then integrate over the eigenvectors. Let $Y=\diag(y_1,\ldots,y_N)$, then we have 
  \begin{equation}\label{tranistion-w-U}
p_{t,s}(y_1,\ldots,y_N\vert x_1,\ldots,x_N)
 \propto
  \triangle_N(y)^2\prod_{n=1}^Ne^{y_n-x_n} \int_{U(N)}{\det}_+(UYU^{-1}-X)^{N(t-s-1)}d\mu(U),
 \end{equation}
where $\mu$ is the Haar measure on the unitary group $U(N)$. The squared Vandermonde determinant which appears in~\eqref{tranistion-w-U} comes from the Jacobian of our change of variables from $H$ to its eigenvalues and eigenvectors. 

The next step is to perform the integral over the unitary group in~\eqref{tranistion-w-U}. We know from~\cite{KKS2016} that
\begin{equation}
\int_{U(N)}{\det}_+(UYU^{-1}-X)^{N(t-s-1)}d\mu(U)\propto
\frac1{\triangle_N(y)\triangle_N(x)}
\det_{1\leq k,\ell\leq N}\Big((y_k-x_\ell)_+^{N(t-s)-1}\Big) 
\end{equation}
where
\begin{equation}
 (a)_+=
 \begin{cases}
 a & \text{if } a>0\\
 0 & \text{otherwise}
 \end{cases}.
\end{equation}
Using standard determinant manipulations, we see that transition probability density can be written as
\begin{equation}\label{transition-C}
 p_{t,s}(y_1,\ldots,y_N\vert x_1,\ldots,x_N)=C\frac{\triangle_N(y)}{\triangle_N(x)}\det_{1\leq k,\ell\leq N}\Big(\kappa_{t-s}(y_k-x_\ell)\Big),
\end{equation}
where $C$ is some constant of proportionality. Thus, it only remains to show that $C=1/N!$.

To find $C$ we will use that the transition probability satisfy
\begin{multline}
 p_{t,s}(y_1,\ldots,y_N\vert x_1,\ldots,x_N)=\\
 \int_{\R^N} p_{t,u}(y_1,\ldots,y_N\vert z_1,\ldots,z_N)p_{u,s}(z_1,\ldots,z_N\vert x_1,\ldots,x_N)dz_1\cdots dz_N
\end{multline}
for $t>u>s>0$. Using~\eqref{transition-C} together with Andreief's integral identity~\cite{Andreief1886,Forrester2018}, we see that
\begin{multline}
 \int_{\R^N} p_{t,u}(y_1,\ldots,y_N\vert z_1,\ldots,z_N)p_{u,s}(z_1,\ldots,z_N\vert x_1,\ldots,x_N)dz_1\cdots dz_N=\\
 N!C^2\det_{1\leq k,\ell,N}\Big(\int_\R \kappa_{t-u}(y-z)\kappa_{u-s}(z-x)dz\Big).
\end{multline}
We recall that $\kappa_{t-s}(x)=w_{N(t-s)-1,1}(x)$. The weight function~\eqref{weight} is identical to the PDF of a gamma distributed random variable and it is well-known that the density have following convolutive property
\begin{equation}\label{kappa-convolution}
\int_\R  \kappa_{t-u}(y-z)\kappa_{u-s}(z-x)dz=\kappa_{t-s}(y-x)
\end{equation}
for $t>u>s>0$ (this is easily verified by direct integration). Thus, using~\eqref{transition-C} once again we get
\begin{multline}
  \int_{\R^N} p_{t,u}(y_1,\ldots,y_N\vert z_1,\ldots,z_N)p_{u,s}(z_1,\ldots,z_N\vert x_1,\ldots,x_N)dz_1\cdots dz_N=\\
 N!Cp_{t,s}(y_1,\ldots,y_N\vert x_1,\ldots,x_N)
\end{multline}
and consequently $C=1/N!$, which completes the proof.
\end{proof}

\begin{corollary}
Let $n\geq 0$ and $t_n>\cdots>t_1>t_0>0$ be integers.
The spatiotemporal joint PDF for the eigenvalues of the LUP at times $t_n,\ldots,t_1,t_0$ is
\begin{multline}\label{jpdf-st}
\wp_{t_n,\ldots,t_0}(x^{(n)}_1,\ldots,x^{(n)}_N;\ldots;x^{(0)}_1,\ldots,x^{(0)}_N)=
\frac{\triangle_N(x^{(n)})\triangle_N(x^{(0)})}{(N!)^nZ_{N(t_0-1),1}(N)}\prod_{j=1}^Nw_{N(t_0-1),1}(x_j^{(0)})\\
\times \prod_{m=1}^n \det_{1\leq k,\ell\leq N}\Big(\kappa_{t_m-t_{m-1}}\big(x_k^{(m)}-x_\ell^{(m-1)}\big)\Big),
\end{multline}
where
\begin{equation}
\big(x_1^{(n)},\ldots,x_N^{(n)},t_n\big),\ldots,\big(x_1^{(0)},\ldots,x_N^{(0)},t_0\big)\in\R^N\times\N
\end{equation}
are space-time points.
\end{corollary}

\begin{proof}
The spatiotemporal joint PDF for the eigenvalues of the LUP is given by
\begin{multline}
\wp_{t_n,\ldots,t_0}(x^{(n)}_1,\ldots,x^{(n)}_N;\ldots;x^{(0)}_1,\ldots,x^{(0)}_N)=
\wp_{t_0}(x^{(0)}_1,\ldots,x^{(0)}_N)\\
\times\prod_{m=1}^n p_{t_m,t_{m-1}}\Big(x^{(m)}_1,\ldots,x^{(m)}_N\Big\vert x^{(m-1)}_1,\ldots,x^{(m-1)}_N\Big),
\end{multline}
where $\wp_s$ is the joint PDF for the eigenvalues of the LUP at time $s$ and $p_{t,s}$ is the transition probability density from time $s$ to time $t$. The corollary follows immediately from~\eqref{jpdf} and~\eqref{tranistion}.
\end{proof}

\section{LUP as a spatiotemporal determinantal point process}

The purpose of this section is to obtain an explicit expression for the spatiotemporal correlation function for the eigenvalues of the LUP.

\begin{definition}
Let $N\geq k_n,\ldots,k_0\geq 0$ be integers and $k:=k_n+\cdots+k_0$. The $k$-point spatiotemporal correlation function is defined as
\begin{multline}\label{corr}
R_{t_n,\ldots,t_0}^k(x_1^{(n)},\ldots,x_{k_n}^{(n)};\cdots;x_1^{(0)},\ldots,x_{k_0}^{(0)}):=
\prod_{m=0}^n\frac{N!}{(N-k_m)!}\\
\times\int_{\R^{(N-k_n)\times\cdots\times(N-k_0)}}
\wp_{t_n,\ldots,t_0}(x_1^{(n)},\ldots,x_N^{(n)};\cdots;x_1^{(0)},\ldots,x_N^{(0)})
\prod_{j=0}^n\prod_{\ell_j=k_j+1}^Ndx^{(j)}_{\ell_j},
\end{multline}
where $\wp_{t_n,\ldots,t_0}$ spatiotemporal joint PDF at times $t_n>\cdots>t_0>0$.
\end{definition}

It is known from the Eynard--Mehta Theorem~\cite{EM1998} (see also~\cite{BR2005}), that given a spatiotemporal joint PDF with structure~\eqref{jpdf-st}, then there exists a correlation kernel $K_N$ such that the correlation functions~\eqref{corr} can be written as
\begin{equation}\label{corr-det}
R_{t_n,\ldots,t_0}^k(x_1^{(n)},\ldots,x_{k_n}^{(n)};\cdots;x_1^{(0)},\ldots,x_{k_0}^{(0)})=
\det_{\substack{1\leq \ell,m\leq n\\i=1,\ldots,k_\ell\\j=1,\ldots,k_m}}
\Big(K_N(x^{(\ell)}_i,t_\ell\,\vert\, x^{(m)}_j,t_m)\Big).
\end{equation}
It is our main goal to find an explicit expression for this kernel. 
We emphasise that the correlation function~\eqref{corr-det} is invariant under gauge transformations
\begin{equation}
K_N(y,t,\vert\, x,s)\mapsto \frac{g(y,t)}{g(x,s)}K_N(y,t,\vert\, x,s)
\end{equation}
for nonzero functions $g$, thus any expression we give for the correlation kernel should only be understood up to overall gauge transformation. In order to give an expression for the kernel, we must first introduce the concept of bi-orthogonal polynomials. 

Let $t>s>0$ be positive integers, we want to introduce families of monic bi-orthogonal polynomials $\{p_k^{(t,s)}(y)=y^k+\cdots\}_{k=0,1,\ldots}$ and $\{q_\ell^{(t,s)}(x)=x^\ell+\cdots\}_{\ell=0,1,\ldots}$ defined through the relation
\begin{equation}\label{bi-rel}
\int_{\R^2} p_k^{(t,s)}(y)\kappa_{t-s}(y-x)q_\ell^{(t,s)}(x) w_{N(s-1),1}(x)dxdy=h_\ell^{(t,s)}\delta_{k,\ell},
\end{equation}
where $\{h_\ell^{(t,s)}>0\}_{\ell=0,1,\ldots}$ is a family of positive constants.

We define the moments and moment determinant as 
\begin{equation}\label{moment-def}
M_{k,\ell}^{(t,s)}:=\int_{\R^2} y^k\kappa_{t-s}(y-x)x^\ell w_{N(s-1),1}(x)dxdy 
\qquad\text{and}\qquad
D_{n}^{(t,s)}:=\det_{0\leq k,\ell\leq {n}}\Big(M_{k,\ell}^{(t,s)}\Big)
\end{equation}
for $k,\ell,n\geq 0$. If the moment determinant is positive $D_{n}^{(t,s)}>0$ for all $n$, then it is known from the general theory of bi-orthogonal polynomials that such bi-orthogonal polynomials exists, see e.g.~\cite{Borodin1998}. Furthermore, we have 
\begin{equation}\label{pq-moments}
 p_n^{(t)}(y)=\frac{1}{D_{n-1}^{(t,s)}}
\det_{\substack{k=1,\ldots,n\\ \ell=1,\ldots,n-1}}
\left(
\begin{array}{c|c}
M_{k,\ell}^{(t,s)} &
y^k
\end{array}
\right)
,\qquad
 q_n^{(s)}(x)=\frac{1}{D_{n-1}^{(t,s)}}
\det_{\substack{k=1,\ldots,n-1\\ \ell=1,\ldots,n}}
\left(
\begin{array}{c}
M_{k,\ell}^{(t,s)} \\
\hline
x^\ell
\end{array}
\right),
\end{equation}
and
\begin{equation}\label{h-moments}
h_n^{(t,s)}=D_n^{(t,s)}/D_{n-1}^{(t,s)}
\end{equation}
for $k,\ell,n\geq0$ and $D_{-1}^{(t,s)}:=1$.

To give the expression for the correlation kernel, we need two additional families of functions defined as integral transforms of the bi-orthogonal polynomials $p_k^{(t,s)}$ and $q_\ell^{(t,s)}$,
\begin{align}
\label{phi-def}
P_k^{(t,u,s)}(z)&:=
\begin{cases}
\displaystyle\int_{\R} p_k^{(t,s)}(y)\kappa_{t-u}(y-z)dy,\hphantom{w_{N(s-1),1}(x)}  & s\leq u<t \\
\,\hphantom{\displaystyle\int_{\R}} p_k^{(t,s)}(z), & u=t
\end{cases}
\\
\label{psi-def}
 Q_\ell^{(t,u,s)}(z)&:=
\begin{cases}
\displaystyle\int_{\R} \kappa_{u-s}(z-x)q_\ell^{(t,s)}(x) w_{N(s-1),1}(x)dx, & s<u\leq t \\
\,\hphantom{\displaystyle\int_{\R}} q_\ell^{(t,s)}(z) w_{N(s-1),1}(z), & s=u
\end{cases}
\end{align}
for $t\geq u\geq s$.
It follows from our original bi-orthogonality relation~\eqref{bi-rel} that the functions $P_k^{(t,u,s)}$ and $ Q_\ell^{(t,u,s)}$ are bi-orthogonal in the following sense.

\begin{corollary}
Let $t>s>0$ be positive integers and $t\geq u\geq s$, then we have
\begin{equation}\label{phi-psi-ortho}
\int_\R P_k^{(t,u,s)}(z) Q_\ell^{(t,u,s)}(z)dz=h_\ell^{(t,s)}\delta_{k\ell}
\end{equation}
for all $k,\ell=0,1,\ldots$.
\end{corollary}

\begin{proof}
The statement is trivial if either $u=s$ or $u=t$, so we assume that $t>u>s$. We have from~\eqref{phi-def} and~\eqref{psi-def} that
\begin{multline}
\int_\R P_k^{(t,u,s)}(z) Q_k^{(t,u,s)}(z)dz=\\
\int_\R
\Big(\int_{\R} p_k^{(t,s)}(y)\kappa_{t-u}(y-z)dy\Big)
\Big(\int_{\R} \kappa_{u-s}(z-x)q_\ell^{(t,s)}(x) w_{N(s-1),1}(x)dx\Big)
dz.
\end{multline}
It is straightforward to verify that the integrals are interchangeable, thus we have
\begin{multline}
\int_\R P_k^{(t,u,s)}(z) Q_k^{(t,u,s)}(z)dz=\\
\int_{\R^2}
p_k^{(t,s)}(y)\Big(\int_{\R} \kappa_{t-u}(y-z)\kappa_{u-s}(z-x)dz\Big)
q_\ell^{(t,s)}(x) w_{N(s-1),1}(x)
dxdy.
\end{multline}
Now using the convolutive property~\eqref{kappa-convolution}, we get
\begin{equation}
\int_\R P_k^{(t,u,s)}(z) Q_k^{(t,u,s)}(z)dz=
\int_{\R^2}p_k^{(t,s)}(y)\kappa_{t-s}(y-x)q_\ell^{(t,s)}(x) w_{N(s-1),1}(x)dxdy
\end{equation}
and the corollary follows from~\eqref{bi-rel}.
\end{proof}

With all of this notation in place, the Eynard--Mehta theorem~\cite[Theorem 13.1.1]{Mehta2004} tells us that the eigenvalues of the LUP described by the spatiotemporal correlation function~\eqref{corr} with kernel
\begin{equation}
K_N(y,u\,\vert\,x,v)=\sum_{k=0}^{N-1}\frac{ P_k^{(t,u,s)}(y) Q_k^{(t,v,s)}(x)}{h_k^{(t,s)}}
-\mathbbm 1_{v>u}\kappa_{v-u}(x-y),
\end{equation}
where $\mathbbm 1_{v>u}$ is the indicator function so that $\mathbbm 1_{v>u}=1$ for $v>u$ and $\mathbbm 1_{v>u}=0$ for $v\leq u$.

It remains to establish expression for the functions $ P_k^{(t,u,s)}$, $ Q_k^{(t,u,s)}$, and the constants $h_\ell^{(t,s)}$, which will be the topic for the rest of this section.

\begin{proposition}
Let $t>s>0$ be positive integers. The moments and moment determinant defined by~\eqref{moment-def} are given by
 \begin{equation}\label{moments}
M_{k,\ell}^{(t,s)}=\frac{\Gamma(N(s-1)+\ell+1)}{\Gamma(N(s-1)+1)}\frac{\Gamma(N(t-1)+k+\ell+1)}{\Gamma(N(t-1)+\ell+1)}
\end{equation}
and
\begin{equation}\label{moments-det}
D_{n}^{(t,s)}=\prod_{k=0}^{n}\frac{\Gamma(N(s-1)+k+1)\Gamma(k+1)}{\Gamma(N(s-1)+1)}>0
\end{equation}
for $k,\ell,n\geq 0$.
\end{proposition}

\begin{proof}
By definition, we have
\begin{equation}
M_{k,\ell}^{(t,s)}=\frac1{\Gamma(N(t-s))\Gamma(N(s-1)+1)}
\int_0^\infty y^ke^{-y}\int_0^yx^{N(s-1)+\ell}(y-x)^{N(t-s)-1}dxdy
\end{equation}
and we know from~\cite[formula 3.191.1]{GR} that
\begin{equation}
\int_0^yx^{a-1}(y-x)^{b-1}dx=y^{a+b-1}\frac{\Gamma(a)\Gamma(b)}{\Gamma(a+b)}
\end{equation}
for $a,b>0$, so
\begin{equation}
M_{k,\ell}^{(t,s)}=\frac{\Gamma(N(s-1)+\ell+1)}{\Gamma(N(s-1)+1)\Gamma(N(t-1)+\ell+1)}
\int_0^\infty y^{N(t-1)+k+\ell}e^{-y}dy.
\end{equation}
Thus~\eqref{moments} follows by recognising the final integral as a Gamma function.

To prove~\eqref{moments-det}, we first note that
\begin{equation}
D_{n}^{(t,s)}=\prod_{j=0}^n \frac{\Gamma(N(s-1)+j+1)}{\Gamma(N(s-1)+1)\Gamma(N(t-1)+j+1)}
\det_{0\leq k,\ell\leq n}\Big(\Gamma(N(t-1)+k+\ell+1)\Big)
\end{equation}
so we only need to prove that
\begin{equation}\label{det-laguerre}
 \det_{0\leq k,\ell\leq n}\Big(\Gamma(N(t-1)+k+\ell+1)\Big)=
 \prod_{j=0}^n \Gamma(N(t-1)+j+1)\Gamma(j+1).
\end{equation}
We will only sketch the proof of this identity, since~\eqref{det-laguerre} is well-known from the theory of Laguerre polynomials (a very similar identity is used to go from~\eqref{L-p} and~\eqref{L-q} to~\eqref{L-sum} in the definition the monic Laguerre polynomial given below). Using the cofactor expansion together with the multiplicative property $z\Gamma(z)=\Gamma(z+1)$, it is seen that
\begin{equation}
 \det_{0\leq k,\ell\leq n}\Big(\Gamma(N(t-1)+k+\ell+1)\Big)=
 \Gamma(N(t-1)+n+1)\Gamma(n+1)\det_{0\leq k,\ell\leq n-1}\Big(\Gamma(N(t-1)+k+\ell+1)\Big).
\end{equation}
Thus~\eqref{det-laguerre} follows by induction, since the identity~\eqref{det-laguerre} is obviously true for $n=0$.
\end{proof}

\begin{corollary}
Let $t>s>0$ be positive integers. The constants~\eqref{h-moments} are given by
\begin{equation}
h_\ell^{(t,s)}=\ell!\frac{\Gamma(N(s-1)+\ell+1)}{\Gamma(N(s-1)+1)}
\end{equation}
for $\ell\geq 0$ and, thus, independent of the time $t$.
\end{corollary}

\begin{proof}
Follows immediately from~\eqref{h-moments} and~\eqref{moments-det}.
\end{proof}

Before we continue, it is worthwhile to recall the structure of the Laguerre polynomials. The (associated) Laguerre polynomials with index $a>-1$ are in monic normalisation are given by
\begin{align}
 \widetilde L^a_n(x)
&=\prod_{k=0}^{n-1}\frac{1}{\Gamma(k+1)\Gamma(a+k+1)}
\det_{\substack{k=1,\ldots,n\\ \ell=1,\ldots,n-1}}
\Big(
\begin{array}{c|c}
\Gamma(a+k+\ell+1) & x^k
\end{array}
\Big) \label{L-p}\\
&=\prod_{k=0}^{n-1}\frac{1}{\Gamma(k+1)\Gamma(a+k+1)}
\det_{\substack{k=1,\ldots,n-1\\ \ell=1,\ldots,n}}
\Big(
\begin{array}{c}
\Gamma(a+k+\ell+1) \\ \hline x^\ell
\end{array}
\Big) \label{L-q} \\
&=\sum_{k=0}^n\frac{(-1)^{n-k}n!}{(n-k)!}\frac{\Gamma(a+n+1)}{\Gamma(a+k+1)}\frac{x^k}{k!} \label{L-sum}
\end{align}
and by construction they are orthogonal
\begin{equation}\label{L-ortho}
\int_\R \widetilde L^a_k(x)\widetilde L^a_\ell(x)w_{a,1}(x)dx=r_\ell^a\delta_{k,\ell},
\qquad\text{with}\qquad 
r_\ell^a=\ell!\frac{\Gamma(a+\ell+1)}{\Gamma(a+1)}.
\end{equation}
We note that
\begin{equation}\label{r}
 r_\ell^{N(s-1)}=\ell!\frac{\Gamma(N(s-1)+\ell+1)}{\Gamma(a+1)}=h_\ell^{(t,s)}.
\end{equation}
In fact, the connection to the Laguerre polynomial extents much further, as we will see through the two following propositions.

\begin{proposition}
Let $t>s$ be positive integers. The bi-orthogonal polynomials $p_n^{(t,s)}$ and $q_n^{(t,s)}$ are monic Laguerre polynomials with index $N(t-1)$ and $N(s-1)$, respectively. That is
\begin{equation}
p_n^{(t,s)}(y)=\widetilde L_n^{N(t-1)}(y)
\qquad\text{and}\qquad
q_n^{(t,s)}(x)=\widetilde L_n^{N(s-1)}(x),
\end{equation}
where the monic Laguerre polynomial $\widetilde L_n^a$ is given by~\eqref{L-sum}.
\end{proposition}

\begin{proof}
We know that the bi-orthogonal polynomials $p_n^{(t,s)}$ and $q_n^{(t,s)}$ can be expressed in terms of the moments through \eqref{pq-moments}. 
For notational convenience, we introduce a constant
\begin{equation}\label{constant}
 c_\ell^{(t,s)}:=\frac{\Gamma(N(t-1)+\ell+1)\Gamma(N(s-1)+1)}{\Gamma(N(s-1)+\ell+1)}.
\end{equation}
Using~\eqref{constant} in our expression for the moments~\eqref{moments} and moment determinant~\eqref{moments-det}, we get 
\begin{equation}
M_{k,\ell}^{(t,s)}=\frac{\Gamma(N(t-1)+k+\ell+1)}{c_\ell^{(t,s)}}
\qquad\text{and}\qquad
D_{n-1}^{(t,s)}=\prod_{k=0}^{n-1}\frac{\Gamma(N(t-1)+k+1)\Gamma(k+1)}{c_k^{(t,s)}}.
\end{equation}
Inserting these expression into~\eqref{pq-moments} yields
\begin{align}
 p_n^{(t,s)}(y)&=
\prod_{k=0}^{n-1}\frac{c_k^{(t,s)}}{\Gamma(N(t-1)+k+1)\Gamma(k+1)}
\det_{\substack{k=1,\ldots,n\\ \ell=1,\ldots,n-1}}
\left(
\begin{array}{c|c}
\displaystyle\frac{\Gamma(N(t-1)+k+\ell+1)}{c_\ell^{(t,s)}} &
y^k
\end{array}
\right),
\\
 q_n^{(t,s)}(x)&=
\prod_{k=0}^{n-1}\frac{c_k^{(t,s)}}{\Gamma(N(t-1)+k+1)\Gamma(k+1)}
\det_{\substack{k=1,\ldots,n-1\\ \ell=1,\ldots,n}}
\left(
\begin{array}{c}
\displaystyle{\Gamma(N(t-1)+k+\ell+1)}/{c_\ell^{(t,s)}} \\
\hline
x^\ell
\end{array}
\right),
\end{align}
which after a simple manipulation reads
\begin{align}
 p_n^{(t,s)}(y)&=
\prod_{k=0}^{n-1}\frac{1}{\Gamma(N(t-1)+k+1)\Gamma(k+1)}
\det_{\substack{k=1,\ldots,n\\ \ell=1,\ldots,n-1}}
\left(
\begin{array}{c|c}
\Gamma(N(t-1)+k+\ell+1) &
y^k
\end{array}
\right),
\\
 q_n^{(t,s)}(x)&=
 \frac1{c_n^{(t,s)}}
\prod_{k=0}^{n-1}\frac{1}{\Gamma(N(t-1)+k+1)\Gamma(k+1)}
\det_{\substack{k=1,\ldots,n-1\\ \ell=1,\ldots,n}}
\left(
\begin{array}{c}
\Gamma(N(t-1)+k+\ell+1) \\
\hline
c_\ell^{(t,s)}x^\ell
\end{array}
\right).
\end{align}
We note that the dependence on the constant~\eqref{constant} have disappeared from our expression of the polynomial $p^{(t)}_n$, and by comparison with~\eqref{L-p} we see that $p^{(t,s)}_n$ is the Laguerre polynomial with index $N(t-1)$ as we set out to prove. The dependence on the constant~\eqref{constant} is still present in our expression of the polynomial $q^{(t,s)}_n$, which makes evaluation a little trickier. By comparison with~\eqref{L-q}, we notice that $q^{(t,s)}_n$ is related to the Laguerre polynomial by a transformation $x^\ell\mapsto c_\ell^{(t,s)}x^\ell$ and multiplication with an overall constant $1/c_n^{(t,s)}$. It follows that we have
\begin{equation}
 q_n^{(t,s)}(x)=
 \frac1{c_n^{(t,s)}}\sum_{k=0}^n\frac{(-1)^{n-k}n!}{(n-k)!}\frac{\Gamma(N(t-1)+n+1)}{\Gamma(N(t-1)+k+1)}
\frac{c_k^{(t,s)}x^k}{k!}.
\end{equation}
Inserting the expression~\eqref{constant} yields
\begin{equation}
 q_n^{(t,s)}(x)=
\sum_{k=0}^n\frac{(-1)^{n-k}n!}{(n-k)!}\frac{\Gamma(N(s-1)+n+1)}{\Gamma(N(s-1)+k+1)}\frac{x^k}{k!},
\end{equation}
which we recognise as the Laguerre polynomial with index $N(s-1)$. This completes the proof.
\end{proof}

\begin{proposition}
Let $t>s$ be positive integers, and let $t\geq u\geq s$. The bi-orthogonal functions $ P_k^{(t,u,s)}$ and $ Q_\ell^{(t,u,s)}$ are given by
\begin{equation}
 P_k^{(t,u,s)}(z)=\widetilde L_k^{N(u-1)}(z)
\qquad\text{and}\qquad
 Q_\ell^{(t,u,s)}(z)=\frac{r_\ell^{N(s-1)}}{r_\ell^{N(u-1)}}\widetilde L_\ell^{N(u-1)}(z)w_{N(u-1),1}(z),
\end{equation}
where $\widetilde L_k^a$ is the monic Laguerre polynomial with index $a$ as defined by~\eqref{L-sum} and $r_\ell^a$ are the norms which appear in the orthogonality relation for the monic Laguerre polynomials~\eqref{L-ortho}.
\end{proposition}

\begin{proof}
By definition $ Q_\ell^{(t,u,s)}$ is given by~\eqref{psi-def}. Inserting our explicit expression for $\kappa_{u-s}$, $q_\ell^{(t,s)}$, and $w_{N(s-1),1}(x)$ into~\eqref{psi-def} yields
\begin{equation}\label{psi-temp}
 Q_\ell^{(t,u,s)}(z)=
\frac{e^{-z}}{\Gamma(N(u-s))\Gamma(N(s-1)+1)}
\int_0^z (z-x)^{N(u-s)-1}x^{N(s-1)}\widetilde L_\ell^{N(s-1)}(x)dx
\end{equation}
for $z>0$ and zero otherwise. Moreover, we know that~\cite[formula 7.412.1]{GR}
\begin{equation}
\int_0^z (z-x)^{\mu-1}x^{a}\widetilde L_\ell^{a}(x)dx
=\frac{\Gamma(a+\ell+1)\Gamma(\mu)}{\Gamma(a+\mu+\ell+1)}z^{a+\mu}\widetilde L_\ell^{a+\mu}(z)
\end{equation}
for $\mu>0$ and $a>-1$. Using this identity to evaluate the integral in~\eqref{psi-temp} gives
\begin{equation}\label{psi-complete}
 Q_\ell^{(t,u,s)}(z)=\frac{r_\ell^{N(s-1)}}{r_\ell^{N(u-1)}}\widetilde L_\ell^{N(u-1)}(z)w_{N(u-1),1}(z),
\end{equation}
where we used~\eqref{r} and~\eqref{weight}.

It remains to find $P_\ell^{(t,u,s)}$. We will use the bi-orthogonal relation~\eqref{phi-psi-ortho}. Inserting~\eqref{psi-complete} into~\eqref{phi-psi-ortho} gives
\begin{equation}
\int_\R  P_k^{(t,u,s)}(z)L_{\ell}^{N(u-1)}(z)w_{N(u-1),1}(z)dz=r_\ell^{N(u-1)}\delta_{k\ell}.
\end{equation}
This bi-orthogonal relation uniquely determines $P_\ell^{(t,u,s)}$ and we see that it is a monic Laguerre polynomial with index $N(u-1)$, which completes the proof.
\end{proof}

The results of this section is summarised by the following theorem.

\begin{theorem}
Let $t_n>\cdots>t_0$ be positive integers. The eigenvalues of the LUP are described by the correlation function~\eqref{corr-det} with kernel
\begin{equation}\label{kernel-laguerre}
K_N^\textup{Laguerre}(y,t\,\vert\, x,s)=
\sum_{k=0}^{N-1}\frac{\widetilde L_k^{N(t-1)}(y)\widetilde L_k^{N(s-1)}(x)}{r_k^{N(s-1)}}w_{N(s-1),1}(x)
-\mathbbm 1_{s>t}\kappa_{s-t}(x-y),
\end{equation}
where $\mathbbm 1_{t>s}$ is the indicator function, so that $\mathbbm 1_{t>s}=1$ for $t>s$ and $\mathbbm 1_{t>s}=0$ for $t\leq s$
\end{theorem}

For equal time correlations ($t=s$), the second term on the right-hand side in~\eqref{kernel-laguerre} is zero and the kernel becomes
\begin{align}
 K_N^\textup{Laguerre}(y\,\vert\, x)&=K_N^\textup{Laguerre}(y,t\,\vert\, x,t)=
 \sum_{k=0}^{N-1}\frac{\widetilde L_k^{N(t-1)}(y)\widetilde L_k^{N(t-1)}(x)}{r_k^{N(t-1)}}w_{N(t-1),1}(x)\\
 &=
 \begin{cases}
 \displaystyle
 \frac{\widetilde L_N^{N(t-1)}(y)\widetilde L_{N-1}^{N(t-1)}(x)-\widetilde L_{N-1}^{N(t-1)}(y)\widetilde L_{N}^{N(t-1)}(x)}
 {y-x}\frac{w_{N(t-1),1}(x)}{r_N^{N(t-1)}}, & x\neq y\\
 \displaystyle
 \Big(\tfrac d{dx}\widetilde L_N^{N(t-1)}(x)\widetilde L_{N-1}^{N(t-1)}(x)
 -\widetilde L_{N-1}^{N(t-1)}(x)\tfrac d{dx}\widetilde L_{N}^{N(t-1)}(x)\Big)
 \frac{w_{N(t-1),1}(x)}{r_N^{N(t-1)}}, & x= y
 \end{cases} \nn
\end{align}
which is the usual Laguerre Christoffel--Darboux kernel as it must be, since we know that the marginal distribution of the LUP is the LUE. We note that for unequal times ($t\neq s$), the first term on the right-hand side in~\eqref{kernel-laguerre} differ from the usual Laguerre Christoffel--Darboux kernel by the fact that the two Laguerre polynomials have different indices. 

\section{Long-time limit and Dyson Brownian motion}

The main purpose of this section is to consider a scaling limit which arises in the long-time limit. However, before we do so, we must recall the basic structure of Hermite polynomials. In monic normalisation, the Hermite polynomials reads
\begin{equation}\label{hermite-def}
\widetilde H_k(x)=
\sum_{j=0}^{\lfloor k/2\rfloor}\frac{(-1)^jk!}{j!(k-2j)!}\frac{x^{k-2j}}{2^{2j}}
\end{equation}
and they satisfy the orthogonality relation
\begin{equation}\label{hermite-ortho}
\int_\R \widetilde H_k(x)\widetilde H_\ell(x)e^{-x^2}dx=m_k\delta_{k\ell},\qquad
m_k=\sqrt\pi\frac{k!}{2^k}
\end{equation}
for $k,\ell=0,1,2,\ldots$. With this minimal notation out of the way, we are ready to state the main result of this section.

\begin{theorem}\label{thm:laguerre-to-hermite}
Let $N$ be a positive integer, then
\begin{equation}
 \lim_{\gamma\to\infty}\sqrt{N\gamma}K_N^\textup{Laguerre}
 \Big(\sqrt{N\gamma}y+N\gamma t,\gamma t\,\Big\vert\,\sqrt{N\gamma}x+N\gamma s,\gamma s\Big)=
 K_N^\textup{Hermite}(y,t\,\vert\,x,s)
 \end{equation}
 with
 \begin{equation}\label{kernel-hermite}
  K_N^\textup{Hermite}(y,t\,\vert\,x,s)=
  \begin{cases}
  \displaystyle
  +\frac1{\sqrt{2s}}\sum_{k=0}^{N-1}\bigg(\frac ts\bigg)^{\!\frac k2}\,
 \frac{\widetilde H_k\big(\frac y{\sqrt{2t}}\big)\widetilde H_k\big(\frac x{\sqrt{2s}}\big)}{m_k}
 e^{-\frac{x^2}{2s}}, & s\leq t,\\
 \displaystyle
 -\frac1{\sqrt{2s}}\sum_{k=N}^{\infty}\bigg(\frac ts\bigg)^{\!\frac k2}\,
 \frac{\widetilde H_k\big(\frac y{\sqrt{2t}}\big)\widetilde H_k\big(\frac x{\sqrt{2s}}\big)}{m_k}
 e^{-\frac{x^2}{2s}}, & s> t,
  \end{cases}
 \end{equation}
where $\widetilde H_k$ is the monic Hermite polynomial~\eqref{hermite-def} and $m_k$ are the constants appearing the orthogonality relation~\eqref{hermite-ortho}.
\end{theorem}

In order to make the proof of Theorem~\ref{thm:laguerre-to-hermite} more transparent, we will provide asymptotic formulae for the Laguerre polynomials $L_k^a$, the constants $r_k^a$, and the weight function $w_{a,1}$ as three separate lemmas before combining these three results to prove the theorem.  

\begin{lemma}\label{lemma:L2H}
Let $k\geq0$ be a fixed integer and $b\in\R$ a fixed constant, then we have
\begin{equation}\label{L2H}
 \lim_{a\to\infty}\frac1{(2a)^{k/2}}\widetilde L_k^{a+b}\big(\sqrt{2a}\,x+a\big)=\widetilde H_k(x),
\end{equation}
where $\widetilde L_k^{a}$ is the monic Laguerre polynomial~\eqref{L-sum} and $\widetilde H_k$ is the monic Hermite polynomial~\eqref{hermite-def}.
\end{lemma}

\begin{proof}
The lemma is minor modification of a known result~\cite[formula 18.7.26]{NIST} and, as we will see below, it follows as a straightforward generalisation of the original proof in~\cite{Calogero1978}.

Let $\{y_{j,k}(a)\}_{j=1,\ldots,k}$ and $\{x_{j,k}\}_{j=1,\ldots,k}$ denote the zeros of the $k$-th Laguerre polynomial with index $a$ and $k$-th Hermite polynomial, respectively. That is
\begin{equation}
\widetilde L_k^{a}(y_{j,k}(a))=0
\qquad\text{and}\qquad
\widetilde H_k(x_{j,k})=0
\end{equation}
for $j=1,\ldots,k$. We have
\begin{equation}\label{L-limit}
\frac1{(2a)^{k/2}}L_k^{a+b}\big(\sqrt{2a}\,x+a)
=\frac1{(2a)^{k/2}}\prod_{j=1}^k\big(\sqrt{2a}\,x+a-y_{j,k}(a+b)\big)
=\prod_{j=1}^k\Big(x-\frac{y_{j,k}(a+b)-a}{\sqrt{2a}}\Big).
\end{equation}
Moreover, we know from~\cite{Calogero1978} that
\begin{equation}\label{zeroes-limit}
y_{j,k}(a+b)=a+\sqrt{2a\,}x_{j,k}+\frac13(1+2k+2x_{j,k}^2+3b)+O(a^{-1/2}), \qquad a\to\infty.
\end{equation}
Combining~\eqref{L-limit} and~\eqref{zeroes-limit}, we see that the left-hand side of~\eqref{L2H} is a polynomial with the same zeros as the Hermite polynomial and leading coefficient one. In other words, it is the monic Hermite polynomial which is the statement of the lemma.
\end{proof}

\begin{lemma}\label{lemma:r-limit}
Let $k\geq0$ be a fixed integer and $b\in\R$ a fixed constant, then we have
\begin{equation}
\lim_{a\to\infty}\frac{r_k^{a+b}}{a^k}=k!,
\end{equation}
where $r_n^{a+b}$ is the squared norms from~\eqref{L-ortho}.
\end{lemma}

\begin{proof}
We know from~\cite[formula 5.11.13]{NIST} that
\begin{equation}
 r_k^{a+b}=k!\frac{\Gamma(a+b+k+1)}{\Gamma(a+b+1)}=k!a^k(1+O(a^{-1}))
\end{equation}
from which the lemma follows.
\end{proof}

\begin{lemma}\label{lemma:weight-limit}
Let $b\in\R$ be a fixed constant, then we have
\begin{equation}\label{weight-limit}
\lim_{a\to\infty}\sqrt{a} w_{a+b,1}(\sqrt ax+a)=\frac{e^{-x^2/2}}{\sqrt{2\pi}},
\end{equation}
where $w_{a,1}$ is the weight function~\eqref{weight}.
\end{lemma}

\begin{proof}
Rather than proving the limit~\eqref{weight-limit} directly, we will prove the limit of the Fourier transform, i.e.
\begin{equation}\label{weight-limit-fourier}
\lim_{a\to\infty}\sqrt{a} \int_\R w_{a+b,1}(\sqrt ax+a)e^{-i\xi x}dx=e^{-\xi^2/2},
\end{equation}
from which our original statement~\eqref{weight-limit} follows. We know that
\begin{equation}
 \sqrt{a} \int_\R w_{a+b,1}(\sqrt ax+a)e^{-i\xi x}dx
 =e^{i\xi\sqrt a}\int_\R w_{a+b,1}(y)e^{-i\xi y/\sqrt a}dy=e^{i\xi\sqrt a}\Big(1+\frac{i\xi}{\sqrt{a}}\Big)^{-a-1}
\end{equation}
and
\begin{equation}
\Big(1+\frac{i\xi}{\sqrt{a}}\Big)^{-a-1}=e^{-i\xi\sqrt a-\xi^2/2}(1+O(a^{-1/2}))
\end{equation}
for $a\to\infty$. This implies~\eqref{weight-limit-fourier} and thereby~\eqref{weight-limit}. 
\end{proof}

\begin{proof}[Proof of Theorem~\ref{thm:laguerre-to-hermite}]
Combining Lemma~\ref{lemma:weight-limit}, \ref{lemma:L2H} and~\ref{lemma:r-limit}, we have
\begin{multline}\label{kernel-L2H}
 \lim_{\gamma\to\infty}\sqrt{N\gamma}K_N^\textup{Laguerre}
 \Big(\sqrt{N\gamma}y+N\gamma t,\gamma t\,\Big\vert\,\sqrt{N\gamma}x+N\gamma s,\gamma s\Big)=\\
 \frac1{\sqrt{2s}}\sum_{k=0}^{N-1}\bigg(\frac ts\bigg)^{\!\frac k2}\,
 \frac{\widetilde H_k\big(\frac y{\sqrt{2t}}\big)\widetilde H_k\big(\frac x{\sqrt{2s}}\big)}{m_k}e^{-x^2/2s}
 -\mathbbm 1_{s>t}\frac{e^{-(y-x)^2/2(s-t)}}{\sqrt{2\pi(s-t)}}.
\end{multline}
Thus, it remains to show that the right-hand side in this equation is equal to the kernel~\eqref{kernel-hermite}. For $s\leq t$, the second term on the right-hand side in~\eqref{kernel-L2H} is zero and the statement follows. Next we recall Mehler's formula~\cite[formula 18.18.28]{NIST}
\begin{equation}
 \sum_{k=0}^\infty\frac{\widetilde H_k(y)\widetilde H_k(x)}{2^{-k}k!} z^k 
 =\frac1{\sqrt{1-z^2}}\exp\Big(\frac{2xyz-(x^2+y^2)z^2}{1-z^2}\Big), \qquad z<1.
\end{equation}
Using Mehler's formula with $z=\sqrt{t/s}<1$, we get
\begin{equation}
 \frac1{\sqrt{2s}}\sum_{k=0}^\infty\bigg(\frac ts\bigg)^{\!\frac k2}\,
 \frac{\widetilde H_k\big(\frac y{\sqrt{2t}}\big)\widetilde H_k\big(\frac x{\sqrt{2s}}\big)}{m_k} 
 =\frac{e^{-(y-x)^2/2(s-t)}}{\sqrt{2\pi(s-t)}}e^{-x^2/2s}
\end{equation}
for $s>t$. Using this expression in~\eqref{kernel-L2H} completes the proof.
\end{proof}

The correlation kernel~\eqref{kernel-hermite} is (up to a trivial gauge transformation) equal to the so-called extended Hermite kernel (see~\cite[formula (3.112)]{Katori2016}), which describes ($\beta=2$) Dyson Brownian motion. Thus, the eigenvalues of the LUP becomes (a version of) Dyson Brownian Motion in the scaling limit described above. Furthermore, we know from studies of the extended Hermite kernel (see~\cite{Katori2016} Section~3.9.2 and~3.9.3 for details) that
\begin{equation}
\lim_{N\to\infty}\frac{e^{-y^2/4t}}{e^{-x^2/4s}}
K_N^\textup{Hermite}\Big(y,2t+\frac N{\pi^2}\,\Big\vert\,x,2s+\frac N{\pi^2}\Big)
=K^\textup{sine}(y,t\,\vert\,x,s)
\end{equation}
and
\begin{equation}
\lim_{N\to\infty}\frac{e^{-y^2/4t}}{e^{-x^2/4s}}
K_N^\textup{Hermite}\Big(y+x_*(t),t+N^{1/3}\,\Big\vert\,x+x_*(s),s+N^{1/3}\Big)
=K^\textup{Airy}(y,t\,\vert\,x,s)
\end{equation}
with
\begin{equation}
x_*(t)=2N^{1/3}+N^{1/3}t-\tfrac14t^2,
\end{equation}
where $K^\textup{sine}$ and $K^\textup{Airy}$ are the extended sine and Airy kernels given by
\begin{equation}\label{sine-extended}
K^\textup{sine}(y,t\,\vert\,x,s)
=
\begin{cases}
\displaystyle+\int_0^1 e^{-\pi^2u^2(s-t)}\cos\pi u(y-x)du, & s\leq t,\\
\displaystyle-\int_1^\infty e^{-\pi^2u^2(s-t)}\cos\pi u(y-x)du, & s>t,
\end{cases}
\end{equation}
and
\begin{equation}\label{airy-extended}
K^\textup{Airy}(y,t\,\vert\,x,s)
=
\begin{cases}
\displaystyle+\int_0^\infty e^{u(s-t)/2}\Ai(y+u)\Ai(x+u) du, & s\leq t,\\
\displaystyle-\int_{-\infty}^0 e^{u(s-t)/2}\Ai(y+u)\Ai(x+u) du, & s>t,
\end{cases}
\end{equation}
respectively. We note that the extended kernels~\eqref{sine-extended} and~\eqref{airy-extended} reduce to the usual sine and Airy kernel for equal time $t=s$.

\paragraph*{Acknowledgement:} We thank Peter Forrester for comments on a first draft of this paper. The author acknowledge financial support by ARC Centre of Excellence for Mathematical and Statistical frontiers (ACEMS).

\appendix

\section{Proof of addition property for the LUE}
\label{appendixA}

The purpose of this appendix is to prove Theorem~\ref{thm:sum}. Our proof is based the `Fourier transform approach' from~\cite{KR2016}; this approach is closely related to a result from~\cite{FR2005} as pointed out in~\cite{FILZ2019}. A statement about sums of LUE random matrices was also presented in~\cite{Kumar2015} (but without a complete proof) and a statement for real matrices (i.e. the Laguerre Orthogonal Ensemble) dates back to~\cite{Muirhead2009}.

The characteristic function (or Fourier transform) of a $\Her_N$-valued random variable $\mathbf X$ is defined as
\begin{equation}\label{fourier}
\phi_{\mathbf X}(T)=\E_{\mathbf X}[e^{i\tr \mathbf X T}],
\end{equation}
where the expectation on the right-hand side is with respect to the measure of $\mathbf X$ and $T$ is a $\Her_N$-valued variable. It worth noting that the characteristic function~\eqref{fourier} is nothing but the entry-wise Fourier transform. Since matrix addition is also entry-wise, we have the standard property
\begin{equation}\label{fourier-sum}
\phi_{\mathbf X+\mathbf Y}(T)
=\E_{\mathbf X+\mathbf Y}[e^{i\tr T(\mathbf X+\mathbf Y)}]
=\E_{\mathbf X}[e^{i\tr T\mathbf X}]\E_{\mathbf Y}[e^{i\tr T\mathbf Y}]
=\phi_{\mathbf X}(T)\phi_{\mathbf Y}(T)
\end{equation}
for independent $\Her_N$-valued random variables $\mathbf X$ and $\mathbf Y$. 

\begin{lemma}\label{lemma:fourier-LUE}
Let $\mathbf X$ be a LUE random matrix with PDF $\mathcal P_{a,b}$, then its characteristic function is
\begin{equation}\label{fourier-LUE}
\phi_{\mathbf X}(T)=\frac1{\det(I_N-iT/b)^{N+a}},
\end{equation}
where $I_N$ denotes the $N\times N$ identity matrix.
\end{lemma}

\begin{proof}
By definition, the characteristic function is given by
\begin{equation}
\phi_{\mathbf X}(T)=\int_{\Her_N} \mathcal P_{a,b}(X)e^{i\tr X T}dX
\end{equation}
with $\mathcal P_{a,b}$ given by~\eqref{Wishart}. Denote by $x_1,\ldots,x_N$ the eigenvalues of $X$, then by an eigenvalue decomposition we have
\begin{equation}
\phi_{\mathbf X}(T)=\int_{\R^N} P_{a,b}(x_1,\ldots,x_N)\Big(\int_{U(N)}e^{i\tr \diag(x_1,\ldots,x_N) UTU^{-1}}dh(U)\Big)\prod_{i=1}^Ndx_i,
\end{equation}
where $P_{a,b}$ is given by~\eqref{LUE} and $dh(U)$ is the normalised Haar measure on the unitary group $U(N)$. The unitary integral within the brackets is a Harish-Chandra--Itzykson--Zuber integral~\cite{HC1957,IZ1980}, thus integrating over the unitary group yields 
\begin{equation}
\phi_{\mathbf X}(T)=\frac{i^{N(N-1)/2}\prod_{j=1}^N\Gamma(j)}{\triangle_N(t)}
\int_{\R^N} \frac{P_{a,b}(x_1,\ldots,x_N)}{\triangle_N(x)}\det_{1\leq k,\ell\leq N}(e^{i x_kt_\ell})\prod_{i=1}^Ndx_i,
\end{equation}
where $t_1,\ldots,t_n$ are the eigenvalues of the matrix $T$.
Inserting the expression~\eqref{LUE} for the PDF $\mathcal P_{a,b}$, we get
\begin{equation}
\phi_{\mathbf X}(T)=\frac{i^{N(N-1)/2}}{\triangle_N(t)Z_{a,b}(N)}\prod_{j=1}^N\Gamma(j)
\int_{\R^N} \triangle_N(x)\det_{1\leq k,\ell\leq N}(e^{i x_kt_\ell})\prod_{i=1}^N w_{a,b}(x_i)dx_i.
\end{equation}
The $N$-fold integral in this expression can be replaced by single integral by means of Andreief's integration formula~\cite{Andreief1886,Forrester2018}. We have
\begin{equation}\label{fourier-sum-proof}
\phi_{\mathbf X}(T)=\frac{i^{N(N-1)/2}N!}{\triangle_N(t)Z_{a,b}(N)}\prod_{j=1}^N\Gamma(j)
\det_{1\leq k,\ell\leq N}\Big(\int_\R x^{k-1}e^{i xt_\ell}w_{a,b}(x)dx\Big).
\end{equation}
Recalling the definition of the weight function~\eqref{weight}, we see that
\begin{equation}\label{weight-sum-proof}
 x^{k-1}w_{a,b}(x)=b^{1-k}\frac{\Gamma(a+k)}{\Gamma(a+1)}w_{a+k-1,b}(x).
\end{equation}
Moreover, $w_{a+k-1,b}$ is the PDF of a gamma random variable, and its Fourier transform (i.e. characteristic function) is known from introductory probability,
\begin{equation}\label{weight-fourier-sum-proof}
 \int_\R w_{a+k-1,b}(x)e^{i xt}dx=\frac{1}{(1-it/b)^{a+k}}.
\end{equation}
Now, using~\eqref{weight-sum-proof} and~\eqref{weight-fourier-sum-proof} in~\eqref{fourier-sum-proof} gives 
\begin{equation}
\phi_{\mathbf X}(T)=\frac{i^{N(N-1)/2}}{\triangle_N(t)Z_{a,b}(N)}\prod_{j=1}^N\frac{\Gamma(j)\Gamma(a+j)}{b^{j-1}\Gamma(a+1)}
\det_{1\leq k,\ell\leq N}\Big(\frac1{(1-it_\ell/b)^{a+k}}\Big).
\end{equation}
The final determinant in this expression is propositional to a Vandermonde determinant, we have
\begin{equation}
\det_{1\leq k,\ell\leq N}\Big(\frac1{(1-it_\ell/b)^{a+k}}\Big)=
(-i/b)^{N(N-1)/2}\triangle_N(t)\prod_{j=1}^N\frac1{(1-it_j/b)^{N+a}}
\end{equation}
Exploting this relation together with the explicit expression for the normalisation constant~\eqref{normalisation} yields
\begin{equation}
\phi_{\mathbf X}(T)=\prod_{j=1}^N\frac1{(1-it_j/b)^{N+a}}
\end{equation}
We recall that $t_1,\ldots,t_n$ are the eigenvalues of the matrix $T$, and thus recognise the right-hand side as the determinant in~\eqref{fourier-LUE}, which completes the proof of the lemma.
\end{proof}

Now, let $\mathbf X$ and $\mathbf Y$ be independent LUE random matrices with PDFs $\mathcal P_{a,b}$ and $\mathcal P_{a',b}$, respectively. We know from~\eqref{fourier-sum} and Lemma~\ref{lemma:fourier-LUE} that
\begin{equation}\label{fourier-sum-LUE}
 \phi_{\mathbf X+\mathbf Y}(T)=\frac{1}{\det(I_N-iT/b)^{2N+a+a'}}.
\end{equation}
The PDF for the random matrix $\mathbf Z=\mathbf X+\mathbf Y$ is obtained by taking the inverse Fourier transform. To do so, we note that~\eqref{fourier-sum-LUE} is identical to~\eqref{fourier-LUE} with a shift $a\mapsto a+a'+N$. Thus, it follows immediately that the inverse Fourier transform is $\mathcal P_{a+a'+N,b}$. In other words, if $\mathbf X$ and $\mathbf Y$ are independent random matrices with PDFs $\mathcal P_{a,b}$ and $\mathcal P_{a',b}$, then $\mathbf Z=\mathbf X+\mathbf Y$ is a LUE  random matrix with PDF $\mathcal P_{a+a'+N,b}$, which is the statement of Theorem~\ref{thm:sum}.

~

~

~


\end{document}